\documentclass[11pt,a4paper,reqno]{amsart}%
\usepackage[latin1]{inputenc}
\usepackage{mathrsfs}
\usepackage{dsfont}
\usepackage{hyperref}
\usepackage{amsmath}
\usepackage{amssymb}
\usepackage{amsthm}
\usepackage{amsfonts}
\usepackage{amstext}
\usepackage{amsopn}
\usepackage{amsxtra}
\usepackage{mathrsfs}
\usepackage{dsfont}
\usepackage{graphicx}
\usepackage{color}

\newtheorem{theorem}{Theorem}[section]

\newtheorem{lemma}[theorem]{Lemma}

\newtheorem{assumption}[theorem]{Assumptions}

\numberwithin{equation}{section}

\newcommand{\ii}{\infty}
\newcommand\R{{\ensuremath {\mathbb R} }}
\newcommand\C{{\ensuremath {\mathbb C} }}
\newcommand\N{{\ensuremath {\mathbb N} }}

\renewcommand\phi{\varphi}

\newcommand{\gH}{\mathfrak{H}}
\newcommand{\gS}{\mathfrak{S}}

\newcommand{\wto}{\rightharpoonup}
\newcommand{\cS}{\mathcal{S}}

\newcommand{\cE}{\mathcal{E}}

\newcommand{\cF}{\mathcal{F}}

\newcommand{\F}{\mathcal{F}}

\renewcommand{\epsilon}{\varepsilon}

\DeclareMathOperator{\tr}{{\rm Tr}}

\renewcommand{\geq}{\geqslant}
\renewcommand{\leq}{\leqslant}

\newcommand{\FNL}{F_{\mathrm{NL}}}

\newcommand{\MFf}{\cF ^{\rm mf}}
\newcommand{\MFe}{ F^{\rm mf}}
\newcommand{\MFgam}{ \gamma_{\rm mf}}

\title[]{Bose gases at positive temperature and non-linear Gibbs measures}

\author[M. Lewin]{Mathieu LEWIN}
\address{CNRS \& Universit\'e Paris-Dauphine, CEREMADE (UMR 7534), Place de Lattre de Tassigny, F-75775 PARIS Cedex 16, France} 
\email{mathieu.lewin@math.cnrs.fr}

\author[P.~T. Nam]{Phan Th\`anh NAM}
\address{IST Austria, Am Campus 1, 3400 Klosterneuburg, Austria} 
\email{pnam@ist.ac.at}

\author[N. Rougerie]{Nicolas ROUGERIE}
\address{Universit\'e Grenoble 1 \& CNRS,  LPMMC (UMR 5493), B.P. 166, F-38042 Grenoble, France}
\email{nicolas.rougerie@grenoble.cnrs.fr}

\date{February 2016}

\begin{document}

\begin{abstract}
We summarize recent results on positive temperature equilibrium states of large bosonic systems. The emphasis will be on the connection between bosonic grand-canonical thermal states and the (semi-) classical Gibbs measures on one-body quantum states built using the corresponding mean-field energy functionals. An illustrative comparison with the case of ``distinguishable'' particles (boltzons) is provided. 
\end{abstract}

\maketitle

\setcounter{tocdepth}{2}
\tableofcontents

\section{Bose Einstein condensation in trapped atomic gases}

The year 1995 saw the first observation of Bose-Einstein condensation (BEC) in ultra-cold dilute alkali vapors, 70 years after the phenomenon had been theoretically predicted. This is of course a major triumph of condensed matter physics, celebrated by the 2001 Nobel prize~\cite{CorWie-nobel,Ketterle-nobel}. Twenty years after this achievement, many important questions raised by the experiments remain unresolved, in particular as regards the rigorous derivation of BEC from the first principles of quantum mechanics, i.e. the many-body Schr\"odinger equation.

Full Bose-Einstein condensation is the phenomenon that, below a certain critical temperature $T_c$, (almost) all particles of a bosonic system must reside in a single quantum state of low energy. Ideally one would like to prove the existence of such a temperature and provide an estimate thereof. For an interacting gas this has so far remained out of reach. In particular, for the homogeneous Bose gas, there is still no proof of Bose-Einstein condensation in the thermodynamic limit, even in the ground state. However, this might not be the main question of interest for the description of cold atoms experiments. Indeed, those are performed in magneto-optic traps, which set a fixed length scale to the problem. The gases in which BEC is observed are thus \emph{not} homogeneous, and the thermodynamic limit is not the most physically relevant regime in this context. 

A lot of progress has been achieved in recent years by considering different scaling regimes, more adapted to the case of inhomogeneous systems, e.g. the mean-field and the Gross-Pitaevskii limits. For trapped systems, BEC in the ground state and its propagation by the many-body Schr\"odinger equation is now fairly well understood (see~\cite{BenPorSch-15,Golse-13,Lewin-ICMP,LieSeiSolYng-05,Rougerie-LMU,Rougerie-cdf,Schlein-08,Seiringer-ICMP10} for reviews). Even in these somewhat more wieldy regimes, very little is known about positive temperature equilibrium states. In particular, an estimate of the critical temperature (or, more generally, temperature regime) is lacking.

In this note, we discuss some of our recent results~\cite{LewNamRou-14d}, in the perspective of the previously mentioned issues. The main idea is to relate, in a certain limit, Gibbs states of large bosonic systems to non-linear Gibbs measures built on the associated mean-field functionals. This is a semi-classical method, where the BEC phenomenon can be recast in the context of a classical field theory. Although, as far as the study of BEC is concerned, our results are rather partial, it is our hope that our methods might in the future help to shed some light on the physics we just discussed. We shall compare the new results to what can be proved in the case of ``boltzons'', i.e. particles with no imposed symmetry, in order to illustrate the crucial importance of Bose statistics in this problem. 

We refer to~\cite{Rougerie-xedp15} for another look at the main results of~\cite{LewNamRou-14d}, with an emphasis on their relation to constructive quantum field theory~\cite{DerGer-13,GliJaf-87,LorHirBet-11,Simon-74,Summers-12,VelWig-73} and the use of invariant measures in the study of non-linear dispersive partial differential equations~\cite{LebRosSpe-88,Bourgain-94,Bourgain-96,Bourgain-97,Tzvetkov-08,BurTzv-08,BurThoTzv-10,ThoTzv-10,CacSuz-14}.  

\vspace{-0.1cm}

\section{From bosonic grand-canonical Gibbs states to non-linear Gibbs measures}

\subsection{Setting}

We consider $N$ bosons living in $\R^d$ and work in the grand-canonical ensemble. Let thus $\gH = L ^2 (\R ^d)$, 
$$\gH ^N = \bigotimes_{\rm sym} ^N \gH \simeq L_{\rm sym} ^2 (\R^{Nd})$$
be the symmetric $N$-fold tensor product appropriate for bosons and
\begin{align*}
\cF &= \C \oplus \gH \oplus \gH ^{2} \oplus \ldots \oplus \gH ^{N} \oplus \ldots\\
\cF &= \C \oplus L ^2 (\R ^d) \oplus L_s ^2 (\R ^{2d})  \oplus \ldots \oplus L_s ^2 (\R ^{Nd}) \oplus \ldots 
\end{align*}
be the bosonic Fock space. We are interested in the positive temperature equilibrium states of the second-quantized Hamiltonian $\mathbb{H}_\lambda$ defined as 
$$\mathbb{H}_\lambda = \mathbb{H}_0 + \lambda \mathbb{W} = \bigoplus_{n=1} ^\infty H_{n,\lambda},$$
with
\begin{align*}
\mathbb{H}_0 &= \bigoplus_{n= 1} ^\infty \left( \sum_{j=1} ^n h_j\right) = \bigoplus_{n= 1} ^\infty \left( \sum_{j=1} ^n -\Delta_j + V (x_j) - \nu \right)\\
\mathbb{W} &= \bigoplus_{n= 2} ^\infty \left( \sum_{1 \leq i<j \leq n}  w_{ij} \right).
%= \bigoplus_{n= 2} ^\infty \left( \sum_{1 \leq i<j \leq n}  w (x_i-x_j) \right)
\end{align*}
Here $\nu$ is a chemical potential, $V$ is a trapping potential, i.e.
$$ V(x) \to + \infty \mbox{ when } |x| \to \infty$$
and $w$ is a \emph{positive}, symmetric, self-adjoint operator on $\gH ^2$. The methods of~\cite{LewNamRou-14d} are limited to rather smooth repulsive interactions, thus $w$ will in general \emph{not} be a multiplication operator. We shall comment on this issue below but, for the time being, think of a finite-rank $w$ with smooth eigenvectors, corresponding to a regularization of a physical interaction. 

We are interested in the asymptotic behavior of the grand-canonical Gibbs state at temperature $T$
\begin{equation}\label{eq:GC Gibbs}
 \Gamma_{\lambda,T} = \frac{1}{Z_\lambda (T)}\exp\left( - \frac1T \mathbb{H}_\lambda\right). 
\end{equation}
The partition function $Z_\lambda (T)$ fixes the trace equal to $1$ and satisfies 
\begin{equation}\label{eq:GC part}
 - T \log Z_\lambda (T)  = F_\lambda (T) 
\end{equation}
where $F_\lambda (T)$ is the infimum of the free energy functional 
\begin{equation}\label{eq:free ener GC}
\F_{\lambda} [\Gamma] = \tr_{\cF} \left[ \mathbb{H}_\lambda  \Gamma \right] + T \tr_{\cF} \left[ \Gamma \log \Gamma \right] 
\end{equation}
over all grand-canonical states (trace-class self-adjoint operators on $\cF$ with trace~$1$). It turns out that an interesting limiting behavior emerges in the regime 
\begin{equation}\label{eq:regime}
 T \to \infty, \quad \lambda = \frac{1}{T}, \quad \nu \mbox{ fixed,} 
\end{equation}
provided one makes the following assumptions:

\begin{assumption}[\textbf{One-body hamiltonian}]\label{asum h}\mbox{}\\
We assume that, as an operator on $\gH = L^2 (\R ^d)$, 
$$ h:= -\Delta + V - \nu > 0$$ 
and that there exists  $p>0$ such that 
\begin{equation}\label{eq:asum trace}
\tr_{\gH} h ^{-p} < \infty. 
\end{equation}
\end{assumption}

Note that one can always find a $p$ such that~\eqref{eq:asum trace} holds. The easiest case, for which our results are the most satisfying, is that where one can take $p=1$ (refered to as the trace-class case). This happens only in 1D and if the trapping potential grows sufficiently fast at infinity, i.e.
\begin{equation}\label{eq:anharm oscill}
 h = -\frac{d^2}{dx^2} + |x| ^a - \nu
\end{equation}
with $\nu$ small enough and $a>2$.

\begin{assumption}[\textbf{Interaction term}]\label{asum w}\mbox{}\\
We pick $w$ a positive self-adjoint operator on $\gH ^2$ and distinguish two cases: 
\begin{itemize}
 \item either one can take $p=1$ in~\eqref{eq:asum trace} and then we assume 
\begin{equation}\label{eq:asum w 1}
 \tr_{\gH ^2} \left[ w\, h ^{-1} \otimes h ^{-1} \right] < \infty
\end{equation}
\item or $p>1$ in~\eqref{eq:asum trace} and we make the stronger assumption that
\begin{equation}\label{eq:asum w 2}
 0 \leq w \leq h ^{1-p'} \otimes h ^{1-p'}
\end{equation}
for some $p'>p$.
\end{itemize}
\end{assumption}

In essence, these (rather restrictive) asumptions correspond to asking that the non-interacting Gibbs state has a well-controled interaction energy. Indeed, one can compute that its two-body density matrix behaves as $T^2 h ^{-1} \otimes h ^{-1}$ in the limit $T\to \infty$. In the 1D case~\eqref{eq:anharm oscill} where $p=1$, we can take for example $w$ a multiplication operator, e.g. by a bounded function $w(x-y)$.  The assumption we make when $p>1$ does not cover such operators.

\subsection{Non-linear Gibbs measures}

The natural limiting object in the setting we just described turns out to be the non-linear Gibbs measure on one-body quantum states given formally by 
 \begin{equation}
d\mu(u)=\frac{1}{Z}e^{-\cE[u]}\,du,
\label{eq:mu_intro}
 \end{equation}
where $Z$ is a partition function and $\cE [u]$ is the mean-field energy functional
\begin{equation}\label{eq:MF func}
\cE [u] := \left\langle u | -\Delta + V - \nu | u \right\rangle_{\gH} + \frac{1}{2} \left\langle u \otimes u | w | u \otimes u \right\rangle_{\gH^2}.
\end{equation}
The rigorous meaning of~\eqref{eq:mu_intro} is given by the two following standard lemmas/definitions:

\begin{lemma}[\textbf{Free Gibbs measure}]\label{lem:free}\mbox{}\\
Write
$$ h = -\Delta + V - \nu = \sum_{j = 1} ^{\infty} \lambda_j |u_j\rangle \langle u_j|$$
and define the asociated scale of Sobolev-like spaces 
\begin{equation}
\gH^s:=D(h^{s/2})=\bigg\{u=\sum_{j\geq1} \alpha_j\,u_j\ :\ \|u\|_{\gH^{s}}^2:=\sum_{j\geq1}\lambda_j^s|\alpha_j|^2<\ii\bigg\}.
\end{equation}
Define a finite dimensional measure on $\mathrm{span} (u_1,\ldots,u_K)$ by setting
$$d\mu_0 ^K (u) := \bigotimes_{j=1} ^K \frac{\lambda_j}{\pi} \exp\left( -\lambda_j |\langle u, u_j\rangle | ^2\right) d \langle u, u_j\rangle$$
where $d\langle u, u_j\rangle = da_j db_j$ and $a_j,b_j$ are the real and imaginary parts of the scalar product.

Let $p>0$ be such that~\eqref{eq:asum trace} holds. Then there exists a unique measure $\mu_0$ over the space $\gH ^{1-p}$ such that, for all $K >0$, the above finite dimensional measure $\mu_{0,K}$ is the cylindrical projection of $\mu_0$ on $\mathrm{span} (u_1,\ldots,u_K)$. 
Moreover 
\begin{equation}\label{eq:DM free meas}
\gamma_0^{(k)}:=\int_{\gH^{1-p}} |u^{\otimes k}\rangle\langle u^{\otimes k}|\;d\mu_0(u) = k!\,(h^{-1})^{\otimes k} 
\end{equation}
where this is seen as an operator acting on $\bigotimes_{\rm sym} ^k \gH$.
\end{lemma}

Note that the free Gibbs measure \emph{never} lives on the energy space $\gH ^1$. It lives on the original Hilbert space $\gH ^0 = \gH$ if and only if~\eqref{eq:asum trace} holds with $p=1$.  We can now define the interacting Gibbs measure as being absolutely continuous with respect to the free Gibbs measure: 

\begin{lemma}[\textbf{Interacting Gibbs measure}]\label{lem:inter}\mbox{}\\
Let 
$$ \FNL [u] := \frac{1}{2}\left\langle u \otimes u | w | u \otimes u \right\rangle_{\gH^2}.$$
If Assumptions~\ref{asum w} hold we have that $ u \mapsto  \FNL [u] $ is in $L ^1 (\gH ^{1-p},d\mu_0)$. In particular 
\begin{equation}\label{eq:int meas}
\mu (du) = \frac{1}{Z_r} \exp\left( - \FNL [u] \right) \mu_0 (du)  
\end{equation}
makes sense as a probability measure over $\gH ^{1-p}$. That is, the relative partition function satisfies 
$$Z_r = \int \exp\left( - \FNL [u] \right) \mu_0 (du) >0  .$$
\end{lemma}

Equation~\eqref{eq:int meas} is the correct interpretation of the formal definition~\eqref{eq:mu_intro}. It is this object that we derive from the bosonic  grand-canonical Gibbs state.

\subsection{Mean-field/large temperature limit}

Our main result in~\cite{LewNamRou-14d} relates, in the limit~\eqref{eq:regime}, the grand-canonical Gibbs state~\eqref{eq:GC Gibbs} to the classical Gibbs measure on one-body state defined in Lemma~\ref{lem:inter}:

\begin{theorem}[\textbf{Derivation of nonlinear Gibbs measures}]\label{thm:main2}\mbox{}\\
Under Assumptions~\ref{asum h} and~\ref{asum w} we have 
$$\frac{F_\lambda (T) - F_0 (T)}{T} \underset{T\to \infty}{\longrightarrow}  -\log Z_r $$
where $F_\lambda (T)$ is the infimum of the free-energy functional~\eqref{eq:free ener GC} and $Z_r$ the relative partition function defined in Lemma~\ref{lem:inter}. 

Let furthermore $\Gamma_{\lambda,T} ^{(k)}$ be the reduced $k$-body density matrix of $\Gamma_{\lambda,T}$. We have, 
\begin{equation}
\frac{\Gamma_{\lambda,T}^{(1)}}{T} \underset{T\to \infty}{\longrightarrow}\int_{\gH^{1-p}} |u\rangle\langle u|\,d\mu(u)
\label{eq:limit_DM_S_p}
\end{equation}
strongly in the Schatten space $\gS^p(\gH)$. In case $p=1$ in~\eqref{eq:asum trace} we also have, for any $k\geq 2$,
\begin{equation}\label{eq:result DM}
\frac{k!}{T^k}\Gamma_{\lambda,T}^{(k)} \underset{T\to \infty}{\longrightarrow} \int_\gH |u^{\otimes k}\rangle\langle u^{\otimes k}|\,d\mu(u)
\end{equation}
strongly in the trace-class.
\end{theorem}

An obvious caveat of our approach is the case $p>1$ where the result is not as strong as one would hope. We conjecture that~\eqref{eq:result DM} continues to hold strongly in the Schatten space $\gS^p(\gH)$ also for $k>1$, but this remains an open problem. More importantly, it would be highly desirable to go beyond the stringent assumption~\eqref{eq:asum w 2}. There are however well-known obstructions to the construction of the Gibbs measure $\mu$ in this case. A minima, a Wick renormalization~\cite{GliJaf-87,Simon-74} must be performed in order to make sense of the measure when $w= w (x-y)$ is a multiplication operator. The derivation from many-body quantum mechanics of the so-defined measure is the subject of ongoing work by the authors. 

We refer to~\cite{LewNamRou-14d,Rougerie-xedp15} for additional comments. Here, let us only discuss the relevance of this result to the BEC phenomenon:
\begin{itemize}
\item First, this theorem seems to be the first giving detailed information on the limit of the thermal states of a Bose system at relatively large temperature. The result~\eqref{eq:limit_DM_S_p} clearly shows the absence of full BEC in the regime we consider. For finite dimensional bosons, things simplify a lot and versions of Theorem~\ref{thm:main2} were known before, see~\cite{Gottlieb-05} and~\cite[Appendix B]{Rougerie-LMU}.
\item The asymptotic regime~\eqref{eq:regime} should be thought of as a mean-field limit. In fact, when $p=1$,~\eqref{eq:limit_DM_S_p} indicates that the expected particle number behaves as $O(T)$, so that taking $\lambda = T ^{-1}$ corresponds to the usual mean-field scaling where the coupling constant scales as the inverse of the particle number. Roughly speaking we are thus dealing with the regime
$$ N \to \infty, \quad \lambda = \frac{1}{N}, \quad T \sim N.$$
\item Still in the case $p=1$, one should expect from a natural extrapolation of this theorem that full BEC does occur in the regime 
$$ N \to \infty, \quad \lambda = \frac{1}{N}, \quad T \ll N$$
which would indicate that the critical temperature scales in this case as the particle number. A natural idea would be to confirm this by studying the concentration of the non-linear Gibbs measure $\mu$ on the mean-field minimizer when the chemical potential is varied.
\end{itemize}

\section{The case of boltzons: derivation of a mean-field free energy functional}

We now change gears and consider, for comparison, the situation where the quantum statistics of the particles is ignored. In this case, we can work in the canonical ensemble, and this is what we shall do here. Let thus $\gH = L ^2 (\R ^d)$ and 
$$\gH ^N = \bigotimes ^N L ^2 (\R ^d) \simeq L ^2 (\R ^{dN})$$
be the $N$-body Hilbert space for $N$ distinguishable particles. We consider the $N$-body free-energy functional at temperature $T$
\begin{equation}\label{eq:bolt free ener func}
\cF_{N,T} [\Gamma_N] := \tr_{\gH ^N} \left[ H_N \Gamma_N \right] + T \tr_{\gH ^N} \left[ \Gamma_N \log \Gamma_N \right] 
\end{equation}
where, to fix ideas, we take
\begin{equation}\label{eq:bolt HN}
H_N := \sum_{j=1} ^N \left( -\Delta_j + V (x_j) \right) + \frac{\lambda}{N-1} \sum_{1\leq i < j \leq N} w_{ij}.
\end{equation}
We make standard asssumptions ensuring that this can be realized as a self-adjoint operator, e.g. that $w = w (x-y) $ is a multiplication operator by a radial function decaying at infinity $w\in L ^p + L ^{\infty}$ with $\max (1,d/2)< p <\infty$. We consider the case of a (say smooth) trapping potential $V(x) \to \infty$ when $|x| \to \infty$, so that $-\Delta + V$ has compact resolvent. 

In~\eqref{eq:bolt free ener func} $\Gamma_N$ is a $N$-body state, that is a trace-class self-adjoint operator on $\gH ^N$ with trace~$1$. For a Hamiltonian such as~\eqref{eq:bolt HN} it is well-known~\cite[Chapter~3]{LieSei-09} that the minimizer at $T= 0$ is bosonic, i.e. 
\begin{equation}\label{eq:bos sym}
 U_{\sigma} \Gamma_N = \Gamma_N U_{\sigma} = \Gamma_N 
\end{equation}
for any permutation $\sigma$ of $N$ indices, where $U_\sigma$ is the associated unitary operator on $\gH ^N$:
$$ \left( U_{\sigma} \Psi_N \right) (x_1,\ldots,x_N) = \Psi_N (x_{\sigma(1)},\ldots,x_{\sigma (N)}).$$
This is no longer true at positive temperature, and in fact we are going to see that in this case, unrestricted minimizers are \emph{not} bosonic.

Due to the symetry of $H_N$, the minimizer of~\eqref{eq:bolt free ener func}, i.e. the Gibbs state
\begin{equation}\label{eq:bolt Gibbs}
\Gamma_{N,T} = \frac{\exp \left( - T ^{-1} H_N \right)}{\tr_{\gH ^N} \left[ \exp \left( - T ^{-1} H_N \right)\right]} 
\end{equation}
satisfies however
\begin{equation}\label{eq:bolt sym}
 U_{\sigma} \Gamma_N U_{\sigma} ^* =  \Gamma_N 
\end{equation}
for any permutation, and we shall make use of this fact. 

Since we do not impose any stronger symmetry, a state of the form $\Gamma_N = \gamma ^{\otimes N}$ is admissible for any one-body state $\gamma$. Taking such an ansatz leads to a mean-field free energy functional
\begin{equation}\label{eq:bolt MF func}
 \MFf [\gamma] = \tr_{\gH} \left[ (-\Delta + V) \gamma \right] + \frac{\lambda}{2} \tr_{\gH ^2} \left[ w \, \gamma ^{\otimes 2} \right] + T \tr_{\gH} \left[ \gamma \log \gamma \right].
\end{equation}
Note that the state $\gamma ^{\otimes N}$ always satisfies~\eqref{eq:bolt sym}, but not~\eqref{eq:bos sym} (see e.g.~\cite[Proposition~3]{HudMoo-75}) unless $\gamma=|\psi\rangle \langle \psi|$ is pure. For the bosonic problem, the entropy term in the mean-field functional would thus be $0$ for any admissible trial state.

We denote $\MFe$ the infimum of $\MFf$ over one-body mixed states and $\MFgam$ the minimizer:
\begin{equation}\label{eq:bolt MF ener}
\MFe := \inf \left\{ \MFf [\gamma], \: \gamma \in \gS ^1 (\gH), \gamma = \gamma ^* ,\gamma \geq 0 , \tr_\gH \gamma = 1  \right\} = \MFf [\MFgam]. 
\end{equation}
Note that $\MFgam$ is unique by strict convexity of the functional, and that for $T\neq 0$ one can easily realize that it must be mixed, i.e not a projector, so that 
$$\tr_{\gH} \left[ \MFgam \log \MFgam \right] < 0 .$$
One can in fact show that the above mean-field functional correctly describes particles with no symmetry in a mean-field limit, with fixed temperature $T$:

\begin{theorem}[\textbf{Mean-field limit for boltzons at finite temperature}]\label{thm:boltzons}\mbox{}\\
Let $\Gamma_{N,T}$ be the unrestricted Gibbs state minimizing~\eqref{eq:bolt MF func} and 
$$ F_{N,T} = - T \log \tr_{\gH ^N} \left[ \exp \left( - T ^{-1} H_N \right)\right]$$
the associated free-energy. Define, via a partial trace\footnote{Remark the normalization: $\Gamma_{N,T} ^{(k)}$ has trace $1$.}, 
$$ \Gamma_{N,T} ^{(k)} = \tr_{k+1 \to N} \Gamma_{N,T}$$
the $k$-body density matrix of $\Gamma_{N,T}$. Then, for any fixed $\lambda$ and $T$ we have 
\begin{equation}\label{eq:bolt ener lim}
\frac{F_{N,T}}{N} \underset{N\to \infty}{\longrightarrow} \MFe
\end{equation}
and 
\begin{equation}\label{eq:bolt DM lim}
\Gamma_{N,T} ^{(k)} \underset{N\to \infty}{\longrightarrow} \MFgam ^{\otimes k}
\end{equation}
strongly in the trace-class.
\end{theorem}

Results in this spirit can be found in~\cite{FanSpoVer-80,PetRagVer-89,RagWer-89,Werner-92}. A proof is presented in Appendix~\ref{sec:app} for completeness. As far as BEC is concerned, the above theorem is mainly interesting in as much as it helps emphasizing the role of Bose statistics at positive temperature\footnote{In much the same way as the study of ``bosonic atoms'' illustrates the role of the Pauli principle in true atoms~\cite{BenLie-83}.}:
\begin{itemize}
\item First, the above theorem indicates that the emergence of the non-linear Gibbs measure discussed previously truly requires Bose statistics. 
\item Since at fixed $T\neq 0$, the minimizer of~\eqref{eq:bolt MF func} is mixed, full BEC does not occur for a trapped boltzonic system in the mean-field limit $N\to \infty$, $T$ fixed. In the same regime, bosons do show BEC~\cite{Suto-03} and the temperature has in fact no effect at leading order~\cite[Section~3.2]{LewNamRou-14}.
\item Thus, for trapped boltzons, we see that the critical temperature scales as $O(1)$ when $N\to \infty$ in the mean-field regime. Only when $T\to 0$ at the same time as $N\to \infty$ does full BEC occur in the Gibbs state.
\item The occurence of BEC for boltzons in this limit is in fact a trivial question of (free) energy balance, very different from the statistical effect that leads to BEC for bosons.
\item As a consequence, one should expect that the critical temperature is much higher for bosons than for boltzons. This can be seen from comparing the above theorem and~\cite[Section~3.2]{Suto-03,LewNamRou-14}: $T^{\rm boltzons}_c \sim 1$ in the mean-field limit whereas $T^{\rm bosons}_c \gg 1$.
\item In fact, from the results of~\cite{LewNamRou-14d} we previously discussed, we see that in some situations the critical temperature for bosons could be as high as $O(N)$, much larger as that one obtains by neglecting statistics. 
\end{itemize}

\bigskip

\noindent \textbf{Acknowledgment.} The research summarized here has received funding from the \emph{European Research Council} (ERC Grant Agreement MNIQS 258023), the \emph{People Programme / Marie Curie Actions} (REA Grant Agreement 291734) and from the \emph{ANR} (Projects NoNAP ANR-10-BLAN-0101 \& Mathostaq ANR-13-JS01-0005-01).

\bigskip

\appendix

\section{Boltzons at positive temperature, proof of Theorem~\ref{thm:boltzons}}\label{sec:app}

Here we briefly sketch a proof of Theorem~\ref{thm:boltzons} inspired by the method introduced in~\cite{MesSpo-82} for classical mean-field limits. We denote 
$$ S_N [\Gamma_N] = - \tr_{\gH^N} \left[ \Gamma_N \log \Gamma_N\right]$$
the von Neumann entropy. We shall use some of its well-known properties for which we refer to~\cite{Carlen-10,OhyPet-93}. As noted in the main text, $\gamma ^{\otimes N}$ is an admissible trial state, and since 
$$ S_N \left[ \gamma ^{\otimes N}\right] = N S_1 \left[ \gamma \right]$$
we clearly get the inequality 
\begin{equation}\label{eq:bolt up bound}
 F_{N,T}\leq N \MFf [\MFgam] = \MFe 
\end{equation}
by the variational principle. 

We turn to the proof of a matching free energy lower bound. Passing to a diagonal subsequence we may assume that as $N\to \infty$ and for any $k \geq 0$, 
$$ \Gamma_{N,T} ^{(k)} \wto_* \gamma ^{(k)}$$
in the trace-class. It is easy to obtain the a priori bound 
\begin{equation}\label{eq:bolt a priori}
 \tr_{\gH ^k} \left[ \left(\sum_{j=1} ^k -\Delta_k + V (x_k)\right) \Gamma_{N,T} ^{(k)} \right] \leq C_{T,\lambda,k} 
\end{equation}
where $C_{T,\lambda,k}$ does not depend on $N$. Since $-\Delta + V$ has compact resolvent we deduce that in fact 
$$ \Gamma_{N,T} ^{(k)} \to \gamma ^{(k)}$$
strongly in the trace-class, and thus 
\begin{align}\label{eq:bolt inf state}
\tr_{\gH ^k} \gamma ^{(k)} &= 1\nonumber \\
\tr_{k+1} \gamma ^{(k+1)} &= \gamma ^{(k)}
\end{align}
where $\tr_{k+1}$ denotes a partial trace with respect to one variable. One may thus apply the quantum de Finetti theorem of St\o{}rmer-Hudson-Moody~\cite{Stormer-69,HudMoo-75} (see also~\cite[Section~3.4]{Rougerie-LMU} and~\cite[Remark~5.3]{LewNamRou-14}). This yields a unique Borel probability measure $\mu$ over the set 
$$ \cS (\gH) = \left\{ \gamma \in \gS ^1 (\gH), \gamma = \gamma ^*, \tr \gamma = 1  \right\}$$
of one-body mixed states such that, for any $k\geq 0$, 
\begin{equation}\label{eq:bolt deF}
\gamma ^{(k)}  = \int_{\cS} \gamma ^{\otimes k} d\mu (\gamma).
\end{equation}
One can easily pass to the liminf in the energy terms of~\eqref{eq:bolt MF func}, following~\cite[Section~3]{LewNamRou-14}, and obtain
\begin{align}\label{eq:bolt lim ener}
\liminf_{N\to \infty} N ^{-1} \tr_{\gH ^N} \left[ H_N \Gamma_{N,T} \right] &= \liminf_{N\to \infty} \tr_{\gH} \left[(-\Delta +V) \Gamma_{N,T} ^{(1)} \right] + \frac{\lambda}{2} \tr_{\gH^2} \left[w\, \Gamma_{N,T} ^{(2)} \right]\nonumber\\
&\geq \tr_{\gH} \left[(-\Delta +V) \gamma ^{(1)} \right] + \frac{\lambda}{2} \tr_{\gH^2} \left[w\, \gamma ^{(2)} \right]\nonumber\\
&= \int_{\cS} \left( \tr_{\gH} \left[(-\Delta +V) \gamma \right] + \frac{\lambda}{2} \tr_{\gH^2} \left[w\, \gamma ^{\otimes 2} \right]\right) d\mu (\gamma).
\end{align}
For the entropy term, we first use subadditivity of the entropy to write, for any $k\geq 0$ 
\begin{equation}\label{eq:bolt subadd}
 \tr_{\gH ^N} \left[ \Gamma_{N,T} \log \Gamma_{N,T}\right] \geq \left\lfloor \frac{N}{k} \right\rfloor \tr_{\gH ^k} \left[ \Gamma_{N,T} ^{(k)} \log \Gamma_{N,T} ^{(k)} \right] + \tr_{\gH ^m} \left[ \Gamma_{N,T} ^{(m)} \log \Gamma_{N,T} ^{(m)} \right] 
\end{equation}
where  (Euclidean division) 
$$N = k \left\lfloor \frac{N}{k} \right\rfloor + m \mbox{ with } m<k. $$
Next, setting
$$ \gamma_0 = c_0 \exp( -\Delta + V)$$
with $c_0$ a normalization constant we use the positivity of the relative entropy to get
\begin{align*}
 \tr_{\gH ^m} \left[ \Gamma_{N,T} ^{(m)} \log \Gamma_{N,T} ^{(m)} \right] &= \tr_{\gH ^m} \left[ \Gamma_{N,T} ^{(m)} \left( \log \Gamma_{N,T} ^{(m)} -\gamma_0 ^{\otimes m}\right)\right] + \tr_{\gH ^m} \left[ \Gamma_{N,T} ^{(m)} \log \gamma_0 ^{\otimes m} \right] \\
 &\geq \tr_{\gH ^m} \left[ \Gamma_{N,T} ^{(m)} \log \gamma_0 ^{\otimes m} \right]. 
\end{align*}
Using~\eqref{eq:bolt a priori} we find that this is bounded below independently of $N$. Thus, returning to~\eqref{eq:bolt subadd}, one can pass first to the liminf in $N$ and then to the sup in $k$ to get
\begin{equation}\label{eq:bolt lim ent}
 \liminf_{N\to \infty} N ^{-1} \tr_{\gH ^N} \left[ \Gamma_{N,T} \log \Gamma_{N,T} \right] \geq \sup_{k\in \N} \frac{1}{k}  \tr_{\gH ^k} \left[ \gamma ^{(k)} \log \gamma ^{(k)} \right]. 
\end{equation}
We then invoke a quantum analogue of a result of~\cite{RobRue-67}:

\begin{lemma}[\textbf{The mean entropy is affine}]\label{lem:mean ent}\mbox{}\\
Let $\left(\gamma ^{(k)}\right)_{k\in \N}$ be a sequence of $k$-body states satisfying~\eqref{eq:bolt inf state}. Then the mean entropy defined by 
$$ -\sup_{k\in \N} \frac{1}{k}  \tr_{\gH ^k} \left[ \gamma ^{(k)} \log \gamma ^{(k)} \right] = -\lim_{k\to \infty} \frac{1}{k}  \tr_{\gH ^k} \left[ \gamma ^{(k)} \log \gamma ^{(k)} \right]$$
is an affine function of $\left(\gamma ^{(k)}\right)_{k\in \N}$.
\end{lemma}

\begin{proof}
That the limit exists and co\"incides with the sup is a consequence of subadditivity of the entropy. Take now two sequences $\left(\gamma_1 ^{(k)}\right)_{k\in \N}$, $\left(\gamma_2 ^{(k)}\right)_{k\in \N}$ and define 
$$ \gamma ^{(k)} = \frac{1}{2} \gamma_1 ^{(k)} + \frac{1}{2} \gamma_2 ^{(k)}.$$
We clearly have, by concavity of the entropy 
$$  \tr_{\gH ^k} \left[ \gamma ^{(k)} \log \gamma ^{(k)} \right] \leq \frac{1}{2 }  \tr_{\gH ^k} \left[ \gamma_1 ^{(k)} \log \gamma_1 ^{(k)} \right] + \frac{1}{2 }  \tr_{\gH ^k} \left[ \gamma_2 ^{(k)} \log \gamma_2 ^{(k)} \right].$$
On the other hand, the $\log$ being operator monotone, we get
\begin{align*}
\tr_{\gH ^k} \left[ \gamma ^{(k)} \log \gamma ^{(k)} \right] &= \frac{1}{2}  \tr_{\gH ^k} \left[ \left( \gamma_1 ^{(k)}\right) ^{1/2} \log \gamma ^{(k)} \left( \gamma_1 ^{(k)}\right) ^{1/2} \right] 
\\&\quad \quad \quad + \frac{1}{2}  \tr_{\gH ^k} \left[ \left( \gamma_2 ^{(k)}\right) ^{1/2} \log \gamma ^{(k)} \left( \gamma_2 ^{(k)}\right) ^{1/2}\right] \\
&\geq \frac{1}{2}  \tr_{\gH ^k} \left[ \left( \gamma_1 ^{(k)}\right) ^{1/2} \log \left( \frac{1}{2} \gamma_1 ^{(k)} \right) \left( \gamma_1 ^{(k)}\right) ^{1/2} \right] 
\\&\quad \quad \quad + \frac{1}{2}  \tr_{\gH ^k} \left[ \left( \gamma_2 ^{(k)}\right) ^{1/2} \log \left( \frac{1}{2} \gamma_2 ^{(k)} \right) \left( \gamma_2 ^{(k)}\right) ^{1/2}\right]\\
&= \frac{1}{2}  \tr_{\gH ^k} \left[ \gamma_1 ^{(k)} \log \gamma_1 ^{(k)} \right] + \frac{1}{2}  \tr_{\gH ^k} \left[  \gamma_2 ^{(k)}\log  \gamma_2 ^{(k)} \right] - \log 2.
\end{align*}
Dividing the previous inequalities by $k$ and passing to the limit gives upper and lower bounds establishing that the limit functional is affine.
\end{proof}

Combining~\eqref{eq:bolt lim ener},~\eqref{eq:bolt lim ent}, inserting the de Finetti representation~\eqref{eq:bolt deF} and using the above lemma we get
$$ \liminf_{N\to \infty} N ^{-1} \cF_{N,T} [\Gamma_{N,T}] \geq \int_{\cS} \MFf [\gamma] d\mu(\gamma) \geq \MFe$$
because $\mu$ is a probability measure. This is the desired lower bound and thus~\eqref{eq:bolt ener lim} is proved. The convergence of states~\eqref{eq:bolt DM lim} follows by combining with~\eqref{eq:bolt up bound}: the measure $\mu$ must be concentrated on the unique minimizer of $\MFf$.

\hfill\qed

%%%%%%%%%%%%%%%%%%%%%%%%%%%%%%%%%%%%%%%%%%
%%%%%%%%%%%%%%%%%%%%%%%%%%%%%%%%%%%%%%%%%%

% \bibliographystyle{siam}
% \bibliography{/home/rougerie/Travail/Documentation/Bibtex/biblio_NR_Fev16}

\end{document}